\documentclass[11pt]{amsart}
\usepackage{amsmath}
\usepackage{amsfonts}
\usepackage{amssymb}

 \usepackage[numbers]{natbib}

\usepackage{amssymb}
\usepackage{amsmath,mathtools}
\usepackage{amsfonts}
\usepackage{amsthm}
\usepackage{amssymb}
\usepackage{natbib}
\usepackage{fullpage}
\usepackage{enumerate}
\usepackage{titlesec}
\titleformat{\section}[hang]{\Large\bfseries\filright}{\thesection.}{.5em}{}
\titleformat{\subsection}[hang]{\large\bfseries\filright}{}{0em}{}
\titleformat{\subsubsection}[block]{\bfseries}{}{0em}{}

\theoremstyle{definition}
\newtheorem{question*}{Question}

\newcommand{\C}{\mathbb{C}}
\newcommand{\R}{\mathbb{R}}

\newtheorem{definition}{Definition}[section]
\newtheorem{theorem}[definition]{Theorem}
\newtheorem{lemma}[definition]{Lemma}
\newtheorem{corollary}[definition]{Corollary}

\theoremstyle{definition}
\newtheorem{remark}[definition]{Remark}

\newtheorem{example}[definition]{Example}

\begin{document}
\setcitestyle{square}
\title{An extension of Bravyi-Smolin's construction for UMEBS}

\author[Jeremy Levick]{Jeremy Levick}
\address{Institute for Quantum Computing and Department of Mathematics, University of Guelph,
Guelph, ON, Canada}
\email{levickje@uoguelph.ca }

\author[Mizanur Rahaman]{Mizanur Rahaman}
\address{BITS Pilani KK Birla Goa Campus, Goa India}
\email{mizanurr@goa.bits-pilani.ac.in}
\begin{abstract}
We extend Bravyi and Smolin's construction for obtaining  unextendible maximally entangled bases (UMEBs)  from equiangular lines. We show that equiangular real projections of rank more than 1 also exhibit examples of UMEBs. These projections arise in the context of optimal subspace packing in Grassmannian spaces. This generalization yields new examples of UMEBs in infinitely many dimensions of the underlying system. Consequently we find a set of orthogonal unitary bases for symmetric subspaces of matrices in odd dimensions. This finding validates a recent conjecture about the mixed unitary rank of the symmetric Werner-Holevo channel in infinitely many dimensions.
\end{abstract}
\maketitle
\section{Introduction and background}
A set  $S$ of orthogonal product states in a bipartite quantum system $\mathbb{C}^d\otimes\mathbb{C}^{d'}$ is called a UPB (unextendible product basis) if $S$ consists of fewer than $dd'$ vectors which have no product states orthogonal to to each element of the set. UPBs have many interesting properties; for example they exhibit nonlocality without entanglement (\citep{upb2}), they provide examples of bound entangled states (\citep{upb1})and indecomposable linear maps, (\citep{terhal}) and also scenarios where quantum correlations don't violate Bell inequalities (\citep{upb3}). Bravyi and Smolin (\citep{bravyi} generalized the idea of UPBs to unextendible maximally entangled bases (UMEBs). A UMEB  is a set $S$ of orthogonal maximally entangled vectors in  $\mathbb{C}^d\otimes\mathbb{C}^{d}$ consisting of fewer than $d^2$ vectors such that there is no maximally entangled vectors orthogonal to $S$. They showed that  UMEBs can be used to show the existence of states for which 1-copy entanglement of assistance (EoA) is strictly smaller than the asymptotic EoA. Since their introduction, there have been a number of papers constructing UMEBs and highlighting interesting properties of theirs (\citep{umeb1}, \citep{umeb2}, \citep{umeb3}). 

One construction that Bravyi-Smolin used to find examples of UMEBs makes use of a set of equiangular lines in $\mathbb{C}^3$.  In this paper we show that we can generalize this example and we produce examples of UMEBs for infinitely many dimensions. We make use of equiangular subspaces and the concept of subspace packing in Grassmannian spaces.

To begin, we introduce some notation we will use throughout: $\{e_i\}_{i = 1}^d \subseteq \C^d$ are the standard basis vectors in $\C^d$; the set of $m\times n$ complex matrices is $M_{m,n}(\C)$ and when $m=n$ we simply write $M_n(\C)$ or even $M_n$. For $A\in M_{n,m}(\C)$, $A^T$ is the transpose of $A$, and $A^*$ its Hermitian adjoint. Inner products are linear in the second argument as in the physicists' convention, so that $\langle u,v\rangle = u^*v$ where $u,v\in \C^d$ and we think of $\C^d \simeq M_{d,1}(\C)$. Given a matrix $A \in M_{m,n}(\C)$, $a_{ij}$ refers to the entries $\langle e_i,Ae_j\rangle = e_i^*Ae_j$ where $e_i\in \C^m$, $e_j \in \C^n$--that is, lower-case letters associated to a matrix designated by an upper-case matrix are the entries of $A$ unless specified otherwise. We note that in the final section, it becomes more convenient to index vectors and matrices starting from $0$, so that in this section (and this section alone!) we use the convention that $e_0 = (1,0,\cdots )^T$, and so forth, and label the entries of a matrix starting $h_{00}$ rather than $h_{11}$.

The central notion in this paper is that of an unextendible maximally entangled basis. To define these, we need the following: 
\begin{definition} A maximally entangled state in $\C^d \otimes \C^d$ is a vector $u$ such that $$u = \frac{1}{\sqrt{d}}(I\otimes U)\phi$$ for some unitary $U \in\mathcal{U}(d)$, and $\phi :=\sum_{i=1}^d e_i \otimes e_i$.
\end{definition}

It is useful to understand this by means of the Choi-Jamiolkowski isomorphism between $\C^d\otimes \C^d$ and $M_d(\C)$, which acts by 

$$j(u\otimes v) = uv^T$$ extended by linearity. So the tensor $u\otimes \bar{v}$ is mapped to the rank-1 matrix $uv^*$.

It is well-known that the action of $M_d$ on itself by left-multiplication can be represented by 

$$j^{-1}(AX) = (I\otimes A)j^{-1}(X);$$

if we observe that $j(\phi) = I_d$, then 
\begin{align*}(I\otimes U)\phi &=  (I\otimes U)j^{-1}(I_d) \\
& = j^{-1}(U).\end{align*}

That is, a maximally entangled state is simply one that satisfies $j(u) = \frac{1}{\sqrt{d}}U$ for some unitary $U \in \mathcal{U}(d)$. More directly, it is a vector obtained by stacking the columns of an order-$d$ unitary matrix on top of one another and rescaling to a unit vector. 

\begin{definition} A maximally entangled basis (MEB) is an orthonormal basis for some subspace of $V\subseteq \C^d\otimes \C^d$ composed entirely of maximally entangled vectors.
\end{definition}

Equivalently, the image of a MEB under the Choi-Jamiolkowski isomorphism is a set of unitary matrices that, after a common rescaling by $\frac{1}{\sqrt{d}}$, are all trace orthonormal. 

\begin{definition} An MEB $\{u_1,\cdots, u_k\} \subseteq \C^d\otimes \C^d$ is unextendible if there are no maximally entangled vectors in the orthogonal complement: there exists no maximally entangled $v$ such that $$\langle v,u_i\rangle = 0$$ for all $i$.
\end{definition}

Once again, we can push this definition forward into $M_d(\C)$ via $j$ to get that an unextendible maximally entangled basis (UMEB) is equivalent to a set of trace-orthonormal (up to rescaling by $\frac{1}{\sqrt{d}}$) unitaries $U_i$ such that there is no unitary $V \in \mathrm{span}\{U_1,\cdots, U_k\}^{\perp}$. 
\subsection{Contributions of this work}
In this work, we construct new examples of UMEBs in infinitely many dimensions of the underlying system. Thus we add more examples in this interesting topic (see \citep{umeb-new}, \citep{umeb1}, \citep{umeb3}). Our novel approach makes use of the optimal subspace packing in Grassmannian spaces to produce these examples. 
We prove the following result (see Theorem \ref{mainthm}):

\begin{theorem}
For each prime $p \equiv -1 \mod 8$, there exists a UMEB of cardinality $p(p+1)/2$ and the corresponding unitary matrices span the symmetric subspace in $M_p(\C)$.
\end{theorem}

We show that equiangular projections of rank more than 1 provide a set of UMEBs in dimensions $p \equiv -1 \mod 8$. We achieve this by generalizing a construction of Bravyi and Smolin who used equiangular lines to produce a set of UMEBs in $\mathbb{C}^3\otimes \mathbb{C}^3$. We prove that their construction can not be extended in higher dimensions if one only considers equiangular lines. Our construction is based on an interesting use of combinatorics and matrix theory and depends on a combinatorial result of Calderbank, Hardin, Rains, Shor and Sloane (\citep{shor-1}). To the best of our knowledge, the use of the optimal subspace packing in Grassmannian space in this topic is new and thus, our work opens up new avenues to explore. We further apply our result to establish that the mixed unitary rank of the symmetric Werner-Holevo channel is the same as the Choi rank. This result supports, in infinitely many dimensions, a conjecture of Girard et. al (see \citep{girard}) about the mixed unitary rank of this channel. 
\section{UMEBs from Equiangular lines}
\subsection{Revisiting the construction of Bravyi-Smolin:}
In \citep{bravyi}, a UMEB in $\C^3\otimes \C^3$ was constructed from a maximal set of real equiangular lines in $\R^3$. We briefly review the theory of real equiangular lines in $\R^d$.

\begin{definition} A set of vectors in $\mathbb{F}^d$ $\{v_1,\cdots, v_k\}$ are said to be equiangular if 

$$|\langle v_i,v_j\rangle| = \begin{cases} 1 & i = j \\ \beta & i\neq j \end{cases}$$ for some fixed constant $\beta$. 
\end{definition}

Following Godsil, \citep{godsil}, we can easily show that as long as $\beta < 1$, a set of equiangular lines must always yield a linearly independent set of rank-one projections, $\{v_iv_i^*\}_{i=1}^k$, which puts a limit on the number of mutually equiangular vectors, $k$. 

\begin{lemma}\label{independent} Let $\{v_i\}_{i=1}^k$ be a set of mutually equiangular vectors in $\mathbb{F}^d$ that are all distinct, i.e. $\beta := |\langle v_i,v_j\rangle|, i\neq j$ is strictly smaller than $1$. 
Then the rank-one projections $\{v_iv_i^*\}_{i=1}^k$ are linearly independent.
\end{lemma}
\begin{proof}\citep{godsil} Define $Y_j = v_jv_j^*-\beta^2 I_d$; notice that 

\begin{equation} \mathrm{tr}(Y_j v_iv_i^*) = \begin{cases} \beta^2 - \beta^2 = 0 & i\neq j \\ 1 - \beta^2 & i=j\end{cases}.\end{equation} 
 
Suppose that $\sum_{i=1}^k x_i v_iv_i^* = 0$ and take the trace of each side against $Y_j$ for arbitrary $j$ to get

$$0 = x_j(1-\beta^2);$$ thus $x_j = 0$ and since $j$ is arbitrary, the projections are linearly independent. 
\end{proof}

If $\mathbb{F} = \R$, so that the vectors $v_i$ are all real, then $v_iv_i^* = v_iv_i^T$ is a symmetric matrix for each $i$ and $\{v_iv_i^*\}_{i=1}^k$ are a basis for a subspace of the symmetric subspace of $M_d(\C)$, $S_d$, which has dimension $\frac{d(d+1)}{2}$; thus $k \leq \frac{d(d+1)}{2}$. A set of real, equiangular lines of size $\frac{d(d+1)}{2}$ is thus said to be maximal. Equivalently, a maximal set of real equiangular lines is a set of real equiangular vectors whose associated rank-one projections span the symmetric subspace $S_d$. 

Maximal sets of real equiangular lines are known to exist in $d = 3$, $d=7$, and $d = 23$ which are the smallest dimensions in which such sets are possible ($d$ must be an odd integer, either $d=3$ or $d = m^2 -2$ for some integer $m$).

\begin{lemma}\label{commonangle} Let $\{v_i\}_{i=1}^{d(d+1)/2}$ be a maximal set of real equiangular lines in $\C^d$. Then the common angle $\beta = \frac{1}{\sqrt{d+2}}$.
\end{lemma}
\begin{proof} The rank-one projections $\{v_iv_i^*\}$ span the symmetric subspace, and hence contain the identity in their span, i.e., there exist real numbers $x_i$ such that 

$$\sum_{i=1}^{d(d+1)/2} x_i v_iv_i^* = I_d.$$ Using $Y_j$ from Lemma \ref{independent}, we see that

$$x_j(1-\beta^2) = 1-d\beta^2$$ and so in fact we have that 

$$\frac{1-d\beta^2}{1-\beta^2} \sum_i v_iv_i^* = I_d.$$
Taking the trace of either side we get 

$$\frac{1-d\beta^2}{1-\beta^2}\frac{d(d+1)}{2} = d$$

which we can solve to get 

\begin{equation}\label{beta} \beta^2 = \frac{1}{d+2}.\end{equation}

\end{proof}

In $d=3$, the $6$ diagonals of a regular icosahedron form a maximal set of real equiangular lines, and Bravyi and Smolin construct a UMEB from this set in the following way:

\begin{lemma} Let $\{v_i\}_{i=1}^6$ be a set of real equiangular lines in $\C^3$. Then there exists $z \in \C$, $|z|=1$ such that $$U_i := I_3 - (1-z)v_iv_i^*$$ are mutually orthogonal unitary matrices.
\end{lemma}
\begin{proof} Unitarity of $U_i$ follows from $|z|=1$, as

\begin{align*}U_i^*U_i &= I_3 - (1-z)v_iv_i^* - (1-\overline{z})v_iv_i^* + (2-z - \overline{z})v_iv_i^*\\ & = I_3 - v_iv_i^*(z + \overline{z} - 2 + 2 - z - \overline{z})\\
& = I_3.\end{align*}

Mutual orthogonality follows from the fact that $v_i$ are equiangular: 

\begin{align*} \mathrm{tr}(U_i^*U_j) & = \mathrm{tr}(I_3) - (1-z) - (1-\overline{z}) + \beta^2 (2-z -\overline{z}) \\
& = 3 +(z + \overline{z} - 2)(1-\beta^2);\end{align*}
set this equal to $0$ to get that 

\begin{equation} \mathrm{Re}(z) = -\frac{3}{2(1-\beta^2)}+1\end{equation} which is compatible with $|z|=1$ so long as the quantity on the right hand side of the equation has modulus less than $1$. Using Equation \ref{beta}  this becomes

$$-1 \leq -\frac{3}{2(1 - \frac{1}{5})} + 1 \leq 1$$ 
which is true, since the quantity in the middle evaluates to $-\frac{15}{8} + 1  = -\frac{7}{8}$. 
 
\end{proof}

\begin{theorem}\citep{bravyi} The unitaries $U_i$ constructed above are the image under the Choi-Jamiolkowski isomorphism of a UMEB. 
\end{theorem}
\begin{proof} Since $v_iv_i^*$ span the symmetric subspace of $M_3(\C)$, so do $U_i = I_3 - (1-z)v_iv_i^*$, thus we have $6$ orthogonal unitaries that span the symmetric subspace of $M_3$, so the orthogonal complement of their span is the antisymmetric subspace. But there are no unitaries in the antisymmetric subspace of $M_3(\C)$ since if $V^T = - V$, then $\mathrm{det}(V) = \mathrm{det}(V^T) = (-1)^3 \mathrm{det}(V)$ and thus $\mathrm{det}(V) = 0$, but a unitary cannot be singular. 
\end{proof}
\subsection*{Impossibility of generalizing the construction to dimensions bigger than 3:}
A natural question to ask is whether sets of maximal, real, equiangular lines in other odd dimensions also yield UMEBs? Our first result shows that the answer is no.
\begin{theorem} Let $\{v_i\}_{i=1}^{d(d+1)/2} \subseteq \C^d$ be a set of real, equaingular lines with common angle $\beta = \frac{1}{\sqrt{d+2}}$. We can find $z_i$, $|z_i|=1$ such that $U_i = I_d + (1-z_i)v_iv_i^*$ are unitaries and $\mathrm{tr}(U_i^*U_j) = 0$ for $i\neq j$ if and only if $d\leq 3$.
\end{theorem}

\begin{proof} First, observe that unitarity of the $U_i$ so-defined is equivalent to $|z_i| = 1$. 
Next, suppose $z_i = z_j$ for distinct indices $i\neq j$. Call this common point on the circle $z$. Then we need 

$$\mathrm{tr}(U_i^*U_j) = d - (1-z) - (1-\overline{z}) + \frac{1}{d+2}(2-z - \overline{z}) = 0.$$

This rearranges to $d + 2(1-\mathrm{Re}(z))(\frac{1}{d+2} - 1) = 0$ or 

$$-\frac{d(d+2)}{2(d+1)} +1= \mathrm{Re}(z).$$

In order for this to be compatible with $|z| = 1$, we need 

$$-1 \leq \frac{-d(d+2) + 2(d+1)}{2(d+1)} \leq 1$$ or 

$$-2d - 2 \leq -d^2 +2 \leq 2d + 2.$$
The inequality on the left becomes $d^2 -2d - 4 \leq 0$ which has roots at $d =1 \pm \sqrt{5}$; the larger root is at $1 + \sqrt{5} \sim 3.23$, so the inequality is satisfied up to $d=3$, but no further. 

So, if it is possible to build $U_i$, each $z_i$ must be distinct. We now show that this is not possible either. Suppose $|z_i| = |z_j| = 1$, and $z_i = a_i + i b_i$, $z_j= a_j + i b_j$. 

Then the condition we want is 

$$\mathrm{tr}(U_i^*U_j) = d - (1-z_i) - (1-\overline{z_j}) + \frac{1}{d+2}(1 - z_i - \overline{z_j} + z_i\overline{z_j}) = 0.$$
Split this into real and imaginary parts to get 

\begin{align} d -2 + a_i + a_j + \frac{1}{d+2}(1 - a_i - a_j + a_ia_j + b_ib_j) & = 0 \\
b_i - b_j + \frac{1}{d+2}(-b_i + b_j + a_jb_i - a_ib_j) &= 0\label{Imag}.\end{align}
 
Equation \ref{Imag} simplifies to 

$$b_i(d+1 + a_j) - b_j(d+1 + a_i) = 0$$ and it is clear that $a_j = -(d+1)$ is not possible for a modulus-$1$ number, so if $a_i = a_j$, we must have $b_i = b_j$ as well, which we have just shown is not possible. Thus $z_i, z_j$ must have distinct real parts. 

Now, we use the fact that $b_i^2 = 1 - a_i^2$ and similarly for $b_j$ to get that 

$$(1-a_i^2)(d+1 + a_j)^2 = (1-a_j^2)(d+1+a_i)^2$$
so that 

$$ 2(d+1)a_j + a_j^2 - a_i^2(d+1)^2 -2(d+1)a_i^2a_j  =  2(d+1)a_i + a_i^2 - a_j^2 (d+1)^2 - 2a_j^2 (d+1)a_i$$

from which 

$$(a_j-a_i)\bigl[2(d+1) + (a_j + a_i) + (d+1)^2(a_j+a_i) + 2(d+1)a_ia_j\bigr] = 0.$$

We have already seen that $a_j = a_i$ is not allowed, so it must be that 

$$2(d+1) + (a_i+a_j)(1 + (d+1)^2) + 2(d+1)a_ia_j = 0$$ which we can solve for $a_i$ in terms of $a_j$ as 

$$a_i = - \frac{2(d+1) + a_j(1+(d+1)^2)}{(1 + (d+1)^2 + 2(d+1)a_j)}.$$
The problem is  that this fixes each $a_i$ for each $i\neq j$ to be the same, so once again, this is not possible. 
\end{proof}

\section{UMEBs from Equiangular Subspaces}
In this section we consider the next natural generalization of the Bravyi-Smolin example by going from rank-one projections $\{u_iu_i^*\}$ arising from equiangular lines, to higher rank projections $\{P_i\}$ satisfying some sort of equiangularity condition.
In general, there is no restriction on the ranks of the $P_i$, but it is easiest to assume that $\mathrm{rank}(P_i) = \mathrm{tr}(P_i) = r$ for all $i$. Under this assumption, we can do a similar analysis as in the case of equiangular lines (when $r=1$). Equiangular projections of higher ranks naturally occur in the packing of subspaces in Grassmannian spaces. In short, there is an equivalence between projections in $M_d(\C)$ and subspaces $V_i \subseteq \C^d$. Given $r$-dimensional subspaces $\{V_i\}$, each corresponds to a point in the Grassmannian $G(r,d)$, and the chordal distance between two points in the Grassmannian corresponds to the quantity
$$d_C^2(V_i,V_j) = \frac{1}{2}\|P_i - P_j\|_F^2  = r - \mathrm{tr}(P_i^*P_j)$$ where $P_i$ is the projection onto the subspace $V_i$. Thus, equidistant points in the Grassmannian correspond to equiangular projections, and maximal packings of equidistant points in the Grassmannian correspond to maximal sets of equiangular projections. For more exposition on this topic, see \citep{shor-1}
\citep{bella}, \citep{bodmann}.
\begin{definition} A set of projections $\{P_i\}_{i=1}^k \subseteq M_d(\C)$ are said to be equiangular if $\mathrm{tr}(P_iP_j) = \beta$ for some fixed $\beta$ and all $i\neq j$.
\end{definition}
\begin{lemma} A set of common-rank (rank $r$), equiangular projections $\{P_i\}_{i=1}^k$ is linearly independent as long as the common angle satisfies $\beta \neq r$.
\end{lemma}
\begin{proof} As before, we set $Y_j = P_j - \frac{\beta}{r} I_d$, and observe that 

$$\mathrm{tr}(Y_jP_i) = \begin{cases} \beta - \beta = 0 & i\neq j \\
r - \beta & i = j\end{cases}.$$

Thus, if $\sum_{i=1}^k x_i P_i = 0$, taking the trace against $Y_j$ for any $j$ gives $(r-\beta)x_j = 0$ and so each $x_j = 0$.
\end{proof}

\begin{lemma} Suppose $P_i$ are real, equiangular projections in $M_d(\C)$. If the projections span the symmetric subspace,  then $\beta = r\frac{r d + r - 2}{(d+2)(d-1)}$.
\end{lemma}
\begin{proof} The previous Lemma shows that there can be at most $d(d+1)/2$ real, equiangular projections as long as the projections are pairwise distinct. If we saturate this inequality, the $P_i$ must span the symmetric subspace, and hence for some real $x_i$, we have 

$$\sum_i x_i P_i = I_d.$$ 
Taking the trace against $Y_j$ we find that 

$$x_j(r - \beta) = r - d\frac{\beta}{r}$$ for each $x_j$. 

Thus $$\frac{r^2 -d\beta}{r(r-\beta)}\sum_i P_i = I_d$$ and taking traces, 

$$\frac{r^2 - d\beta}{r-\beta}\frac{(d+1)}{2} = 1$$ and so 

$$\beta = r\frac{r d + r - 2}{(d+2)(d-1)}.$$

\end{proof}

Our idea, then, is to mimic the Bravyi-Smolin construction, but using higher-rank real, equiangular projections instead of rank-ones. That is, we wish to construct unitaries 
\begin{equation} U_i = I_d - (1-z_i)P_i\end{equation} with $|z_i| = 1$ and $P_i$ a set of $\frac{d(d+1)}{2}$ real, equiangular projections in $M_d(\C)$ so that $\mathrm{tr}(U_i^*U_j) = 0$. 

\begin{lemma}\label{duality} Suppose there exists a set of $\frac{d(d+1)}{2}$ real, equiangular projections $P_i$ of rank $r$ in $M_d(\C)$ spanning the symmetric subspace. Then there also exists projections $Q_i$ of rank $r' = d-r$ satisfying the same conditions. Moreover, if we can construct unitaries $U_i = I_d - (1-z)P_i$ that are mutually orthogonal, we can do the same for $Q_i$.
\end{lemma}
\begin{proof} Define $Q_i = I_d - P_i$; these are clearly projections of rank $r'$ and we can check that $\mathrm{tr}(Q_iQ_j) = d - 2r + \beta = (d-r) + (\beta - r) = r' + (\beta - r)$. 
Define $\beta' = r' + \beta - r$ so that $\beta' - r' = \beta - r$; since $\beta - r \neq 0$, neither is $\beta' - r'$, so these are linearly independent real projections by Lemma \ref{independent} and so must also span the symmetric subspace.

Now, suppose that $\mathrm{tr}(U_i^*U_j) = d - r(1-z) -r (1-\overline{z}) + \beta (2 - z - \overline{z}) = 0.$ 

Then this becomes $d +(\beta - r)(2- z - \overline{z}) = 0$ and so this has a solution so long as we can find $|z|=1$ such that $\beta - r = - \frac{d}{2(1-\mathrm{Re}(z))}$. But because $\beta' - r' = \beta -r$ any $z$  that works for the projections $P_i$ will also work for the dual projections $Q_i$.
\end{proof}

\begin{theorem} {\label{thm-equiangular}}
Let $\{P_i\}_{i=1}^{d(d+1)/2}$ be a set of real, equiangular projections in $M_d(\C)$ spanning $S_d$, and thus necessarily satisfying 

$$\mathrm{tr}(P_iP_j) = r\frac{rd + r - 2}{(d+2)(d-1)}.$$ 
Then there exists a $z$ such that $$U_i = I_d - (1-z)P_i$$ is unitary and $\mathrm{tr}(U_i^*U_j) = 0$ for all distinct $i,j$ if and only if $2r - 1 \leq d \leq 2r + 1$. 
\end{theorem}
\begin{proof} As before, unitarity of the $U_i$ is equivalent to $|z| = 1$. So we simply check trace-orthogonality:

$$ \mathrm{tr}(U_i^*U_j) = d - (1-z)r - (1-\overline{z})r + r\frac{rd + r - 2}{(d+2)(d-1)}(2 - z - \overline{z})$$ so setting this to zero we get

$$0  = d - r(2 -z - \overline{z})\bigl(1-\frac{rd+r-2}{(d+2)(d-1)}\bigr)$$ which rearranges to

$$1-\mathrm{Re}(z)  = \frac{d(d+2)(d-1)}{2r(d+1)(d-r)}.$$
 
 Thus, $\mathrm{Re}(z) = \frac{2r(d+1)(d-r) - d(d+2)(d-1)}{2r(d+1)(d-r)}$ and this quantity must have modulus less than or equal to $1$, hence 
 
 $$-2r(d+1)(d-r) \leq 2r(d+1)(d-r) - d(d+2)(d-1) \leq 2r(d+1)(d-r).$$
 The right-hand inequality is trivially satisfied, so we need
 
 \begin{equation}\label{Ineq} d(d+2)(d-1) \leq  4r(d+1)(d-r).\end{equation}
 
 We will show now that this inequality only has positive integer solutions for rank/dimension pairs $(r,d)$ when $d \in \left[2r-1,2r+1\right]$. We will do this by expanding out the factors in Inequality \ref{Ineq} to get an inequality for the cubic polynomial in $d$ for any fixed $r$: 
 $$p(d) = d^3 + d^2(1-4r) + d(4r^2 -4r -2) + 4r^2 \leq 0.$$
 
Since $p(d)$ is cubic, it has at most three roots. We will find all three roots, and so know exactly the regions on which $p(d)\leq 0$ and hence \ref{Ineq} is satisfied. 

First, observe that $\lim_{d\rightarrow -\infty} = -\infty \leq 0$, while at $d=0$ we get $p(0) = 4r^2 \geq 0$, so one root, $a(r)$, lies in $(-\infty,0)$. 
 
Next, we show that at $d = 2r-2$ the inequality fails, while at $d = 2r-1$ it is satisfied:
When $d = 2r-2$, $4r = 2(d+2)$, and $d - r = \frac{d-2}{2}$. Thus, Inequality \ref{Ineq} becomes

$$d(d+2)(d-1) \leq 2(d+2)(d+1)\frac{(d-2)}{2}$$ or 

$$(d+2)\bigl[ d^2 -d - (d^2 -d -2)\bigr] = 2(d+2) \leq 0;$$ but for $r>0$, $d=2r-2 > -2$ so this cannot hold. Thus $p(d) > 0$ at $d = 2r-2$.

On the other hand, if $d = 2r-1$ for $r>0$, we get that $4r = 2(d+1)$ and $d-r = \frac{d-1}{2}$. 

Thus we analyze 
$$d(d-1)(d+2) \leq 2(d+1)^2 \frac{(d-1)}{2}$$ or 

$$(d-1)\bigl[d^2 + 2d - (d+1)^2\bigr] = -(d-1)\leq 0;$$ for $r>0$ and $d =2r-1$ this is always true. So $p(d) \leq 0$ for $d = 2r-1$. Thus, there is a root $b(r)$ in $\left(2r-2,2r-1\right)$. 

Finally, we observe that the left hand side of Inequality \ref{Ineq} is independent of $r$, while the right hand side is symmetric under mapping $r \mapsto d-r$, so if $p(d) \leq 0$ for fixed $r$, the same is true for $d-r$, and similarly for if $p(d)\geq 0$. Thus $p(d) \geq 0$ for $d = 2(d-r) -2$ which rearranges to $d = 2r+2$, and $p(d) \leq 0$ for $d = 2(d-r) -1$, or $d = 2r+1$. So a third root $c(r)$ lies in $\left(2r+1,2r+2\right)$. Note that this is essentially an observation about the dual projections: a rank/dimension pair $(r,d)$ admits projections with the desired properties if and only if the pair $(d-r,d)$ does as well. 

This exhausts all possible roots of $p(d)$, so we have that $p(d) \leq 0$ for $d \in \left(-\infty,a(r)\right) \cup \left(b(r),c(r)\right)$, and the intersection of this with the set of positive integers is just $\{2r-1,2r,2r+1\}$. Thus for fixed $r>0$, these are the only values of $d$ satisfying Inequality \ref{Ineq}.
\end{proof}

\begin{remark} Setting $r=1$, we see more clearly why only $d=3$ admits unitaries built from a maximal set of equiangular lines: the previous Theorem shows that this will only work if $1 \leq d \leq 3$. 
\end{remark}

\begin{remark} We can calculate $z$ in each of the three allowable cases, for a fixed $d$. From $\mathrm{Re}(z) = \frac{2r(d+1)(d-r) - d(d+2)(d-1)}{2r(d+1)(d-r)}$ we can substitute in $d = 2r-1, 2r, 2r+1$ to get respectively

\begin{align*} \mathrm{Re}(z) &= \frac{-(d+1)^2 + 2}{(d+1)^2} &= -1 + \frac{2}{(d+1)^2} \\
\mathrm{Re}(z) & = \frac{-d(d+1) +4}{d(d+1)}& = -1 + \frac{4}{d(d+1)}\\
\mathrm{Re}(z) & = \frac{-(d+1)^2 + 2}{(d+1)^2} & = -1 + \frac{2}{(d+1)^2}\end{align*}
where the final one can be calculated explicitly or simplified using the duality between $d = 2r-1$ and $d = 2r+1$. 
\end{remark}

Since the Bravyi-Smolin construction relies on the odd dimension of the space to show that there are no unitaries in the antisymmetric subspace, it makes sense to consider $d = 2k+1$; in this case to find equal-rank equiangular real projections to mimic the Bravyi-Smolin construction, we need the common rank to be $r=k$ or $r=k+1$, and from the duality relation of Lemma \ref{duality} it suffices to consider just the case $r=k$, and get $k+1 = r' = d - r$ from duality. In this case, we solve that $\mathrm{Re}(z) = -1 + \frac{2}{(d+1)^2}$ where $d = 2k+1$, so we have $\mathrm{Re}(z) = -1 + \frac{1}{2(k+1)^2}$. 

Thus, for the rest of the paper, we will consider the possibility of finding $\binom{d+1}{2}$ real, rank-$\frac{d-1}{2}$ projections in $M_d(\C)$ with $\mathrm{tr}(P_iP_j) = r\frac{dr + r - 2}{(d+2)(d-1)} = \frac{d^2 -5}{4(d+2)}$ for odd $d$. Remarkably, as we see below, there exist a set of $\binom{d+1}{2}$ real projections in dimensions $d=p\equiv -1\mod 8$ with these exact properties.  

\subsection*{A construction for $d = p \equiv -1 \mod 8$ by Calderbank, Hardin, Rains, Shor and Sloane}
A quick remark that in this section, we reindex our standard labelling for vectors so that vectors and matrices are now $0$-indexed: i.e., the standard basis vector $(1,0,\cdots)^T$ in $\C^d$ is now $e_0$, and the standard basis is thus $\{e_0,\cdots, e_{d-1}\}$. The same holds for entries of matrices: the entries of a matrix $H$ are $\{h_{ij}\}_{i,j=0}^{d-1}$. This is a notational convenience to match the construction we present below.

In \citep{shor-1}, a construction was given for obtaining mutually orthogonal real projections $P_i\in M_p(\C)$ of rank $\frac{p-1}{2}$ such that

$$\frac{p-1}{2} - \mathrm{tr}(P_iP_j) = \frac{(p+1)^2}{4(p+2)},$$ for $p = 3$, or $p$ a prime satisfying $p\equiv - 1 \mod 8$. Notice that with $d=p$, we have the correct rank $r=\frac{d-1}{2}$, and we can rearrange to see that 

$$\mathrm{tr}(P_iP_j) = \frac{2(p-1)(p+2)- (p+1)^2}{4(p+2)} = \frac{p^2 -5}{4(p+2)}$$ as we require for maximal, real equiangular projections. 
 
We outline the construction for completeness, which requires the existence of a Hadamard matrix $H$ of size $\frac{p+1}{2}$. Denote by $\mathcal{Q}=\{q_1,\cdots, q_{(p-1)/2}\}$ the non-zero quadratic residues modulo $p$, and by $\mathcal{R} = \{r_1,\cdots, r_{(p-1)/2}\}$ the non-residues. If $C = \frac{1+\sqrt{p+2}}{\sqrt{p+1}}$ then an orthogonal basis for the span of $P_t$ is given by the vectors 

$$e_{q_s} + h_{st} C e_{k q_s}$$ as $s,t$ vary as $0 \leq t \leq \frac{p-1}{2}$, $1\leq s\leq \frac{p-1}{2}$ and for some fixed $k\in \mathcal{R}$. Note that although these vectors are not normalized, they all have the same norm (as each has exactly two non-zero entries: a $1$ and a $\pm C$) and so they can be simultaneously normalized to an orthonormal basis. This yields $\frac{p+1}{2}$ projections each of dimension $\frac{p-1}{2}$. 

Further, each $P_i$ also yields $p-1$ more projections by cyclically permuting the entries of each vector; this gives $\frac{p(p+1)}{2}$ projections in total. 
See Theorem 3 in \citep{shor-1} for more details, or \citep{conway} for a specific example with $p=7$. Also see 
\citep{zhang} for some more generalizations on this construction.

Now that we know the existence of $\frac{p(p+1)}{2}$ real projections of rank $\frac{p-1}{2}$, for dimensions $d=p\equiv -1\mod 8$, we can obtain the required UMEB in this dimensions.
\begin{theorem} {\label{mainthm}}
For each $d = p \equiv -1 \mod 8$, there exists a UMEB of cardinality $d(d+1)/2$ and the corresponding unitary matrices span the symmetric subspace in $M_d(\C)$.
\end{theorem}
\begin{proof} The projections constructed in the way  above are compatible with the construction from Bravyi and Smolin. In particular, following Theorem \ref{thm-equiangular} we can find a $z$, $|z|=1$ such that $U_t  = I_p - (1-z)P_t$ are a set of mutually orthogonal unitaries. Since $P_t$ span the symmetric subspace of $M_p(\C)$, any unitary orthogonal to all of the $U_t$ must itself be antisymmetric, and since $p$ is odd, this is impossible. Thus this construction yields a UMEB. 
\end{proof}
 
 We elucidate the construction of Calderbank, Hardin, Rains, Shor and Sloane using the example of $d = p = 7$. 

\begin{example}  The non-zero quadratic residues modulo $7$ are $\mathcal{Q} = \{1,2, 4\}$ and so the non-residues are $\mathcal{R} = \{3,5,6\}$. We have $C = \frac{1+\sqrt{p+2}}{\sqrt{p+1}} =\sqrt{2}$. 

We also need a $4\times 4$ Hadamard matrix, for example
\[H:=\begin{bmatrix} 1 & 1 \\ 1 & -1 \end{bmatrix}^{\otimes 2} = \begin{bmatrix} 1 & 1 & 1 & 1\\ 1 &-1 & 1 & - 1\\ 1 & 1 & -1 & -1 \\ 1 & -1 & -1 & 1 \end{bmatrix}.\]
We fix a $k\in \mathcal{R}$, say $k = 3$. For each $s = 1,2,3$ we have that $q_s = 1,2,4$ respectively, so $P_t$ is the span of the three vectors 

\[ e_{1} + C h_{1t} e_{3}, e_2 + Ch_{2t} e_{6}, e_4 + Ch_{3t} e_{5}\] for $t = 0,1,2,3$. 

So for example when $t=0$ we get that $P_0$ is the span of 

\begin{align*} e_1 + Ch_{10} e_3 &= \begin{pmatrix} 0 & 1 & 0 & C & 0 & 0 & 0 \end{pmatrix}^T \\
 e_2 + Ch_{20} e_6 & = \begin{pmatrix} 0 & 0 & 1 & 0 & 0 & 0 & C\end{pmatrix}^T \\
 e_4 + Ch_{30} e_5 & = \begin{pmatrix} 0 & 0 & 0 & 0 & 1 & C & 0 \end{pmatrix}^T\end{align*}
 
where we use the fact that $h_{10} = h_{20} = h_{30} = 1$ since the $0^{th}$ column of $H$ has all of its entries to be $1$, and we use modular arithmetic so that for example when $q_s = 4$, $k\times q_s = 3\times 4 \equiv 5 \mod 7$. 
 
Similarly, when $t = 1$ we get $P_1$ to be the span of the vectors

\begin{align*}e_1 + Ch_{11} e_3 & = \begin{pmatrix} 0 & 1 & 0 & -C & 0 & 0 & 0 \end{pmatrix}^T \\
e_2 + Ch_{21} e_6 & =\begin{pmatrix} 0 & 0 & 1 & 0 & 0 & 0 & C\end{pmatrix}^T \\
e_4 + Ch_{31} e_5 & = \begin{pmatrix} 0 &0 & 0 & 0 & 1 & -C & 0 \end{pmatrix}^T. \end{align*}

By doing the same for $t=2,3$ we get two more dimension-$3$ subspaces, and to obtain the remaining subspaces, take each $P_t$, and simultaneously cyclically permute the entries of each of the vectors used to define it. So for example, from $P_0$ we also obtain a subspace that is the span of 

\begin{align*} \begin{pmatrix} 0 & 0 & 1 & 0 & C & 0 & 0 \end{pmatrix}^T \\
\begin{pmatrix} C & 0 & 0 & 1 & 0 & 0 & 0\end{pmatrix}^T \\
\begin{pmatrix} 0 & 0 & 0 & 0 & 0 & 1 & C \end{pmatrix}^T
\end{align*} 
by shifting each entry in each vector forward by $1$, and continuing through all $6$ cyclic shifts. Thus, each of the four subspaces $P_0,P_1,P_2,P_3$ is part of a family each of $7$ members (itself and its six cyclic shifts), for a total of $28$ rank-$3$ projections. 

To normalize each vector, we divide by its norm of $\sqrt{1 + \sqrt{2}^2} = \sqrt{3}$ so one may confirm by calculation that for any pair of distinct projections, the trace of the product is 

$$\frac{2(1-C^2)^2 + (1+C^2)^2}{9} = \frac{2(-1)^2 + (1+2)^2}{9} = \frac{11}{9} = \frac{p^2-5}{4(p+2)}$$ as required.
 
\end{example}

\begin{remark} It is worth commenting on the combinatorial interpretation of the construction in \citep{shor-1}. As per \citep{zhang}, the use of the Hadamard matrix relies on the existence of a corresponding $2$-design, or, equivalently, the difference set underlying it. Given any Hadamard matrix of order $\frac{p-1}{2}$, there exists a difference set in the additive group of integers modulo $p$: $D = \{d_1,\cdots, d_{\frac{p-1}{2}}\}$, such that $(-1)D = \{-d_1,\cdots, -d_{\frac{p-1}{2}}\}$ is also a difference set. The vectors in the span of the projections all have the form 

$$e_{d_i + x} + h_{ij}C e_{-d_i + x}$$ for some $x$; $D+x$ and $-D + x$ are lines of two block-designs, and it is the intersection properties of the blocks in these designs that yield equiangularity. 

It is well-known that in the case of rank-one projections, i.e. equiangular lines, examples can be built both from strongly regular graphs with certain appropriate parameters, and from block-designs; e.g. there are $28$ real equiangular lines in $\R^7$ whose pattern of non-zero entries is determined by the Fano plane. It is interesting to explore what other combinatorial data can give rise to equiangular projections, comparable to the case of equiangular lines.
\end{remark}

\begin{remark} The construction from \citep{shor-1} works for $d=7$, $d=23$, the two other dimensions for which a maximal set of real equiangular lines exists--even though the equiangular lines themselves are not amenable to the Bravyi-Smolin's construction. 
\end{remark}

\subsection{On the mixed unitary rank of the symmetric Werner-Holevo channel:}
Note that a quantum channel $\Phi:M_d(\C)\rightarrow M_d(\C)$ is called mixed-unitary if there exists a $r>1$, a probability vector $(p_1,\cdots, p_r)$ and unitary matrices $U_1,\cdots, U_r$ in $M_d(\C)$ such that  for all $X\in M_d(\C)$ we have 
\[\Phi(X)=\sum_{j=1}^r p_j U_jXU_j^*\]
The \textit{mixed-unitary rank} (see \citep{girard}) of a mixed-unitary channel is the minimum number of distinct unitary conjugations needed for the expression above. 

The symmetric Werner-Holevo channel ($WH_+$)  is defined on $M_d(\C)$ as follows 
\[WH_+(X)=\frac{Tr(X)1+X^T}{d+1},\]
where $X^T$ is the transpose of $X$. It is known that this channel is mixed unitary for all $d$. It was shown in \citep{girard} that the mixed-unitary rank of this channel is equal to $d(d+1)/2$ for even $d$ and they conjectured, based on numerical evidence, that the mixed unitary rank of this channel remains to be $d(d+1)/2$ even when $d$ is odd. Following Theorem 19 in \citep{girard}, it is evident that the mixed unitary rank of $WH_+$ is $d(d+1)/2$ if and only if there is an orthonormal unitary basis of the symmetric subspace. Since the unitaries that give rise to UMEBs in Theorem \ref{mainthm} constitute a unitary basis of the symmetric subspace we prove the following result which provide further evidence for the conjecture mentioned in \citep{girard}.
\begin{corollary}
The symmetric Werner-Holevo channel $WH_+(X)=\frac{Tr(X)1+X^T}{d+1}$ on $M_d(\C)$ has mixed unitary rank $d(d+1)/2$ for $d = p \equiv -1 \mod 8$.
\end{corollary}
\begin{proof}
This follows from the Theorem 19 in \citep{girard} and Theorem \ref{mainthm}.
\end{proof}

\section{Conclusion}
In quantum information processing, quantum measurements based on mutually unbiased bases have been proved to be very useful. It turns out (see \citep{MUBs}) that certain UMEBs can be turned into such mutually unbiased bases. Hence finding new examples of UMEBs remains an interesting problem in this field. In this work, we have introduced a new approach to obtain examples of UMEBs in infinitely many dimensions of the underlying system. While the equianguar lines have been very useful in many aspects of quantum computing (\citep{SIC-1},\citep{SIC-2}), the parallel theory of equiangular subspaces relating to quantum computational tasks is not well developed. We demonstrate that optimal equiangular subspaces in certain dimensions are very useful in finding new examples of UMEBs. It has been observed recently that such special subspaces can be used in the framework of quantum error correction (see \citep{isoclinic}). It will be interesting to see how our techniques used here could provide stronger connection between UMEBs and quantum error correction. We also showed that the mixed-unitary-rank conjecture of the symmetric Werner-Holevo channel holds true in infinitely many dimensions. It will be interesting to see how a more general connection between UMEBs and mixed-unitary rank of an arbitrary mixed-unitary channel can be formulated.  
\bibliography{UMEBs}
\bibliographystyle{amsplain}

\end{document}